\documentclass[12pt]{amsart}
\usepackage{amsthm,amsmath,amssymb}
\usepackage{graphicx}
\usepackage{pgf,tikz}
\usepackage{natbib}
\usepackage{placeins}

\newtheorem{lemma}{Lemma}[section]

\theoremstyle{definition}

\begin{document}

\title{Adjustment with Three Continuous Variables}
\author{Brian Knaeble} 
\date{June 2017}
\maketitle
\begin{abstract}
Spurious association between $X$ and $Y$ may be due to a confounding variable $W$.  Statisticians may adjust for $W$ using a variety of techniques.  This paper presents the results of simulations conducted to assess the performance of those techniques under various, elementary, data-generating processes.  The results indicate that no technique is best overall and that specific techniques should be selected based on the particulars of the data-generating process.  Here we show how causal graphs can guide the selection or design of techniques for statistical adjustment.  R programs are provided for researchers interested in generalization.  
\end{abstract}
\section{Introduction}

Statistical adjustment has meant different things to different people over the years.  \citet{Deming64} wrote a book titled Statistical Adjustment of Data, focusing on the use of known facts to lesson errors in measurement.  Epidemiologists routinely make adjustments \citep{Kasim12}, and cancer researchers report incidence and mortality rates after adjusting for variables such as age and ethnicity \citep{Howlader16}.  Whether or not to adjust is a question relevant even to those studying workforce discrimination \citep{Schieder16}. 

The slope coefficient from simple (bivariate) linear regression may be adjusted to account for a third lurking variable \citep{Lu09}.  Multiple regression is one of many ways to make an adjustment.  Stratification over the third variable is another possibility, and there are many varieties of matching procedures, should the original regressor be dichotomous.  When all three variables are continuous adjustment can be accomplished using partial and semi-partial correlation.  

In some instances conditioning on a third variable can increase bias \citep{Pearl09b}.  Given an accurate causal graph (also known as a Directed Acyclic Graph or DAG) the back-door criterion can be used to determine an admissible set of covariates on which to condition for unbiased estimation of a causal effect \citep{Pearl09a}.  Some argue for conditioning on as many pre treatment variables as possible \citep{Rubin09}, but this is not without controversy.  See the introduction of \citet{Ding15} for an overview.  

Here we have conducted numerous simulations, restricting our attention to the simplest case of only three variables.  We assume linearity and continuity, and we study 33 different data generating processes to assess the performance of multiple regression compared with simple regression and methods relating to partial correlation.  For detailed description of assumptions and studied adjustment techniques see Section 2.  For graphical displays of results see Section 3.  Commentary on the results is provided in Section 4.  For further discussion and conclusions see Section 5.
\section{Methods}  
There are 25 non-cyclic graphs of three variables upon which to generate data.  The classic confounding case is shown here: 
\newline
 \begin{tikzpicture}
\node (x) at (-1,1) {$X$};
\node (y) at (1,1) {$Y$};
\node (w) at (0,0) {$W$};
 
\draw[->, line width= 1] (x) -- (y);
\draw[->, line width= 1] (w) -- (x);
\draw[->, line width= 1] (w) -- (y);
\end{tikzpicture}  
Each of the three arrows can be independently modified so as to vanish or reverse direction.  For example, $W$ could be an intermediate variable as shown here:
\newline
 \begin{tikzpicture}
\node (x) at (-1,1) {$X$};
\node (y) at (1,1) {$Y$};
\node (w) at (0,0) {$W$};
 
\draw[->, line width= 1] (x) -- (w);
\draw[->, line width= 1] (w) -- (y);
\end{tikzpicture} 
Most would agree given the latter causal graph that conditioning on $W$ is inappropriate, because a change in $X$ can not produce a change in $Y$ when $W$ is held fixed.  

We write $X\rightarrow Y$ to represent a direct causal effect of $X$ on $Y$.  We use a dash to represent the absence of a direct causal effect, e.g. $X-Y$.  We write $X\leftarrow Y$ when $Y$ causes $X$.  In this latter case the true causal effect of $X$ on $Y$ is zero.  Analogous notation and interpretation holds for relationships with $W$.  $X$ and $Y$ are defined or redefined if necessary so that any causal effect between them is positive not negative.  When there are two causal effects involving $W$, we default to them both being positive, and we allow for the effect between $W$ and $X$ to be positive or negative.  We signal a negative effect with a lower case n, e.g. $X\leftarrow^nW$.  We use the adjective ``twisted'' when the direction of the effect between $W$ and $X$ differs from the direction of the effect between $W$ and $Y$.  We do not allow cycles.  

For consistency, data is generated assuming normality and unit variance, and associated with each arrow is a true slope coefficient of magnitude $1/\sqrt{3}$.  To maintain unit variance for a response variable, defined as any variable at the tip of an arrow, the error term has standard deviation $\sqrt{2}/\sqrt{3}$ when there is one cause, and $1/\sqrt{3}$ for each when there are two causes.  The sample size can be set arbitrarily, and here we have chosen to set the sample size at $n=30$.  

The following R commands,
\begin{verbatim}
x=rnorm(30,0,1)
y=(1/sqrt(3))*x+rnorm(30,0,sqrt(2)/sqrt(3))
w=(1/sqrt(3))*x+(1/sqrt(3))*y+rnorm(30,0,1/sqrt(3))
\end{verbatim}  
generate data for this graph:
\newline
 \begin{tikzpicture}
\node (x) at (-1,1) {$X$};
\node (y) at (1,1) {$Y$};
\node (w) at (0,0) {$W$};
 
\draw[->, line width= 1] (x) -- (y);
\draw[->, line width= 1] (x) -- (w);
\draw[->, line width= 1] (y) -- (w);
\end{tikzpicture}  
In general, following our assumptions, each graph has its own data generating process.

We assume that researchers, blind to the data generating process, will attempt to estimate the causal effect of $X$ on $Y$ by 1. regressing $Y$ on $X$ (Simple Regression), 2. regressing $Y$ on $X$ and $W$ (Multiple Regression), 3. regressing $Y$ on the residuals from a regression of $X$ on $W$ (Residual $X$), 4.  regressing the residuals from a regression of $Y$ on $W$, onto $X$ (Residual $Y$), 5.  regressing the residuals from a regression of $Y$ on $W$, onto the residuals from a regression of $X$ on $W$ (Residual $X$ and $Y$), or 6.  regressing $Y$ onto the fitted values of a regression of $X$ onto $W$ (Fitted $X$).  We show in the appendix that Multiple Regression, Residual $X$, and Residual $X$ and $Y$ are equivalent, each having conditioned on $W$.  The results thus display estimates from only Simple Regression, Multiple Regression, Residual $Y$, and Fitted $X$.  

\section{Results}
For each data generating process, one thousand datasets were produced, and the four adjustment techniques applied to each.  The true, total causal effect is defined from the data generating process, using the chain rule to multiply slope coefficients of a unidirectional causal chain, and adding effects across paths when appropriate.  The true, total causal effect is represented graphically with a vertical black line.  Figures are classified by how $W$ relates to both $X$ and $Y$ within the data generating process, while within-figure-plots are classified by how $X$ relates to $Y$ within the data generating process.

\begin{figure}[p]
\centering
\includegraphics[width=5in]{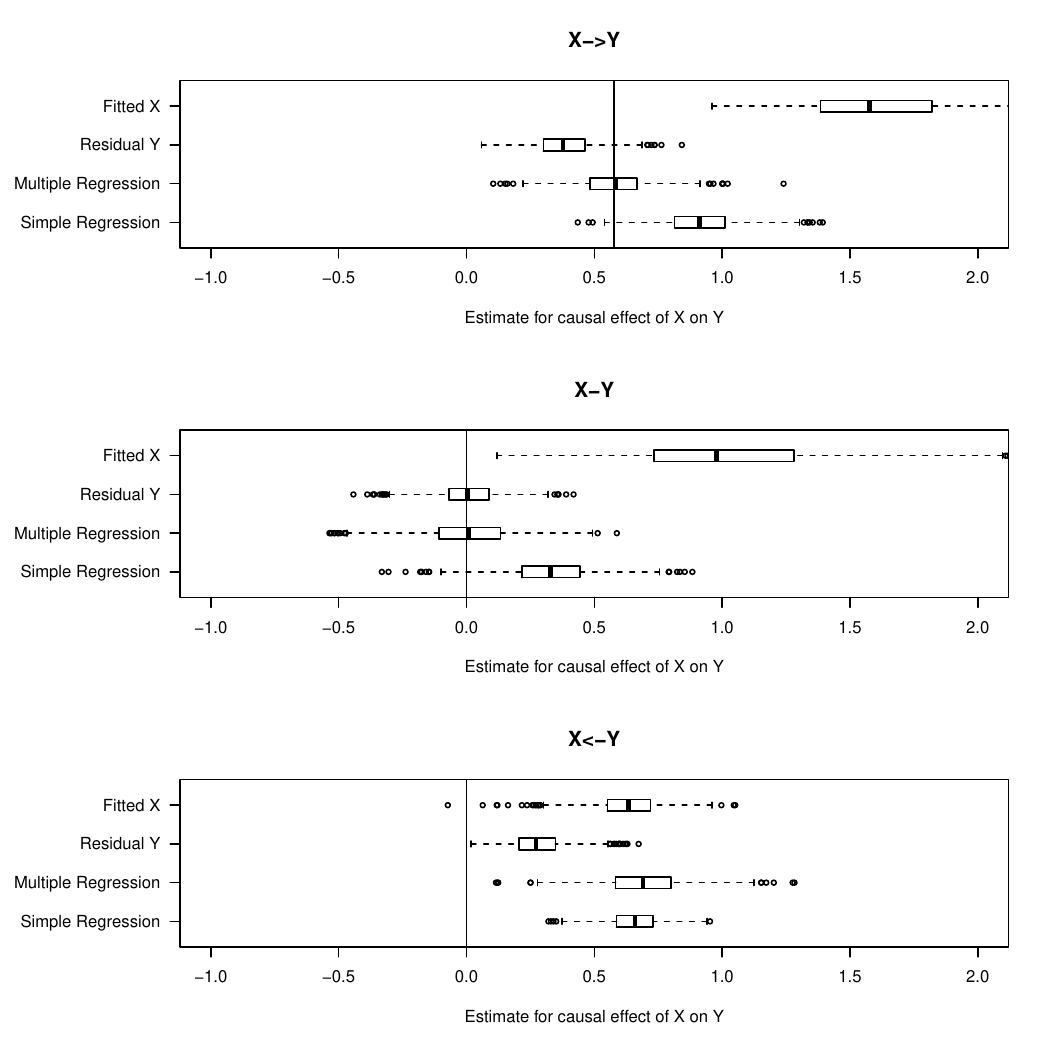}
\caption{Confounding: $X\leftarrow W\rightarrow Y$.}
\label{Conf}
\end{figure}

\begin{figure}[p]
\centering
\includegraphics[width=5in]{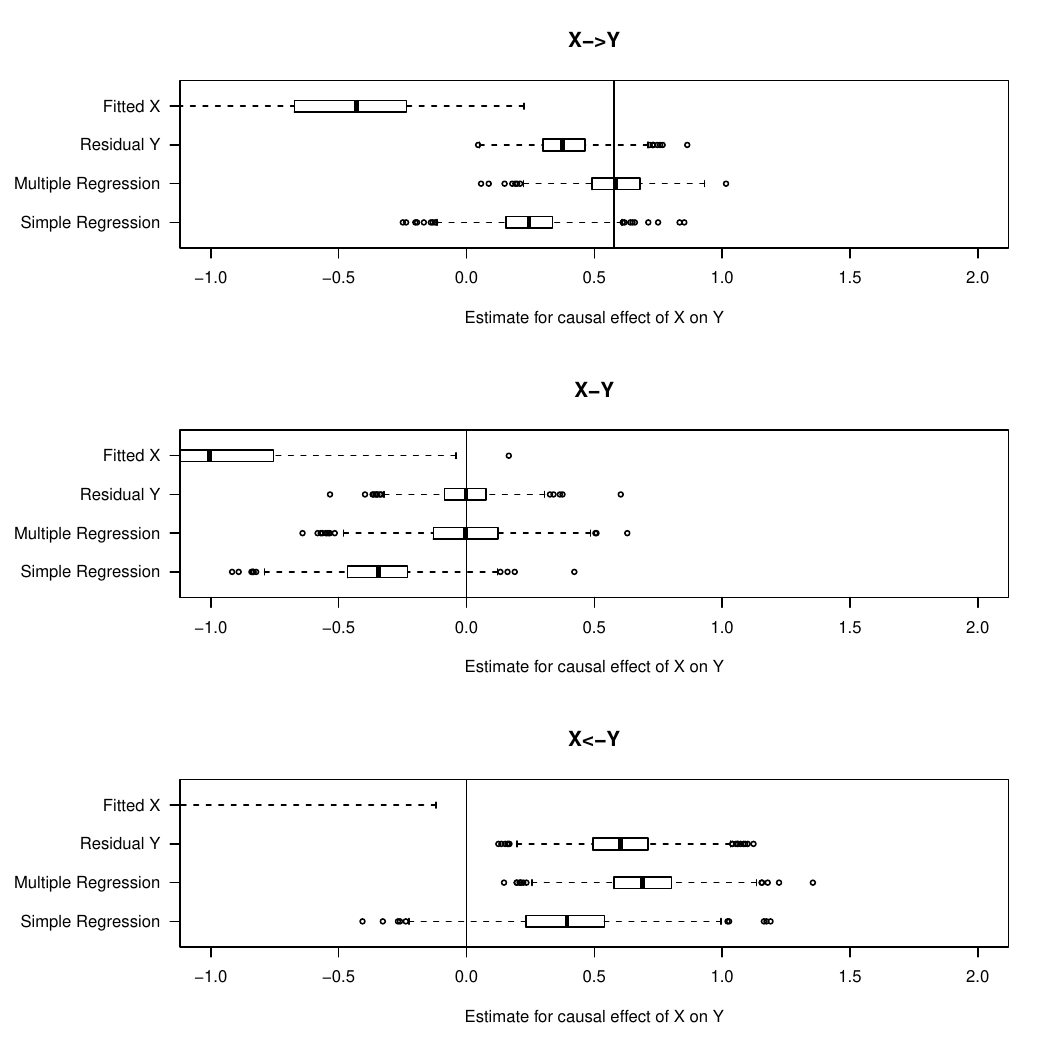}
\caption{Twisted Confounding: $X\leftarrow^n W\rightarrow Y$.}
\label{TConf}
\end{figure}

\begin{figure}[p]
\centering
\includegraphics[width=5in]{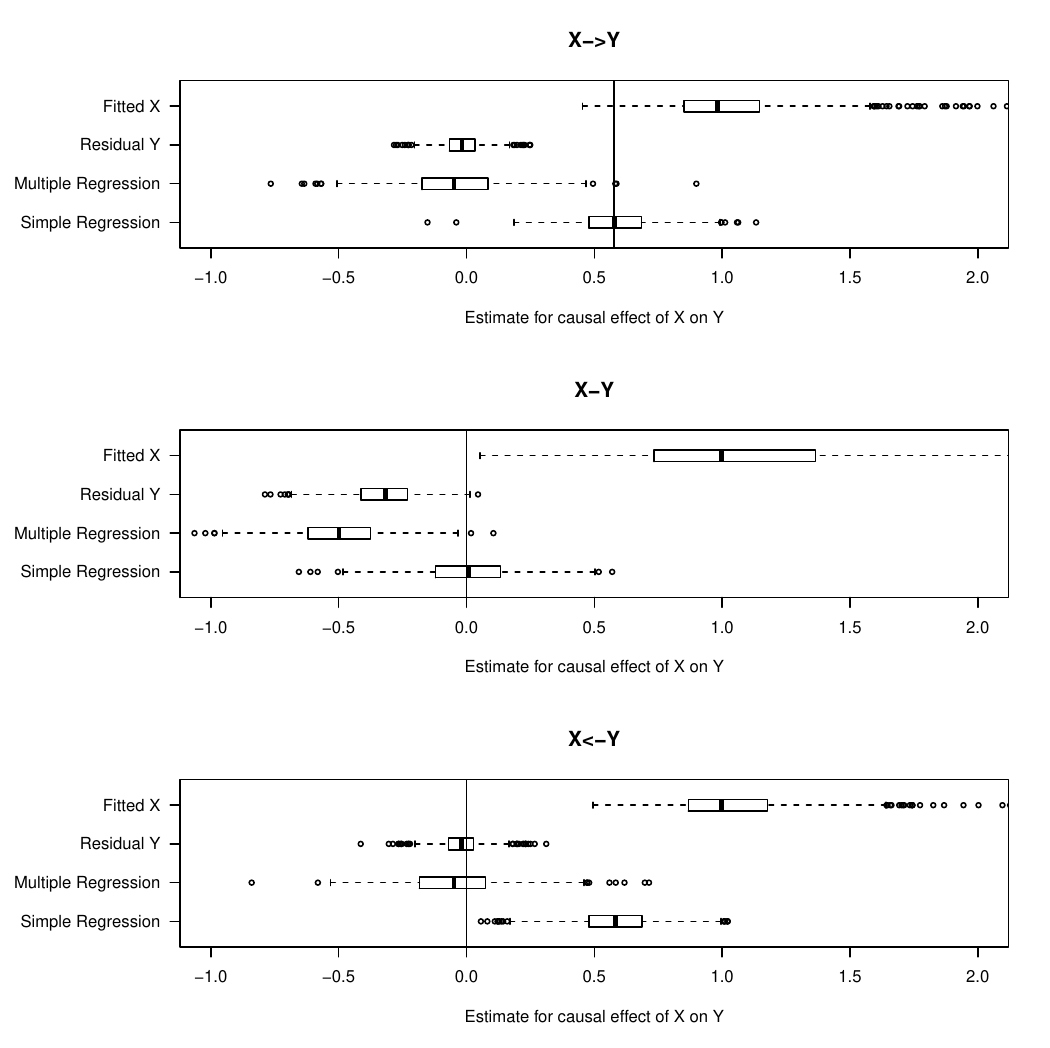}
\caption{$W$ as a collider: $X\rightarrow W\leftarrow Y$.}
\label{Col}
\end{figure}

\begin{figure}[p]
\centering
\includegraphics[width=5in]{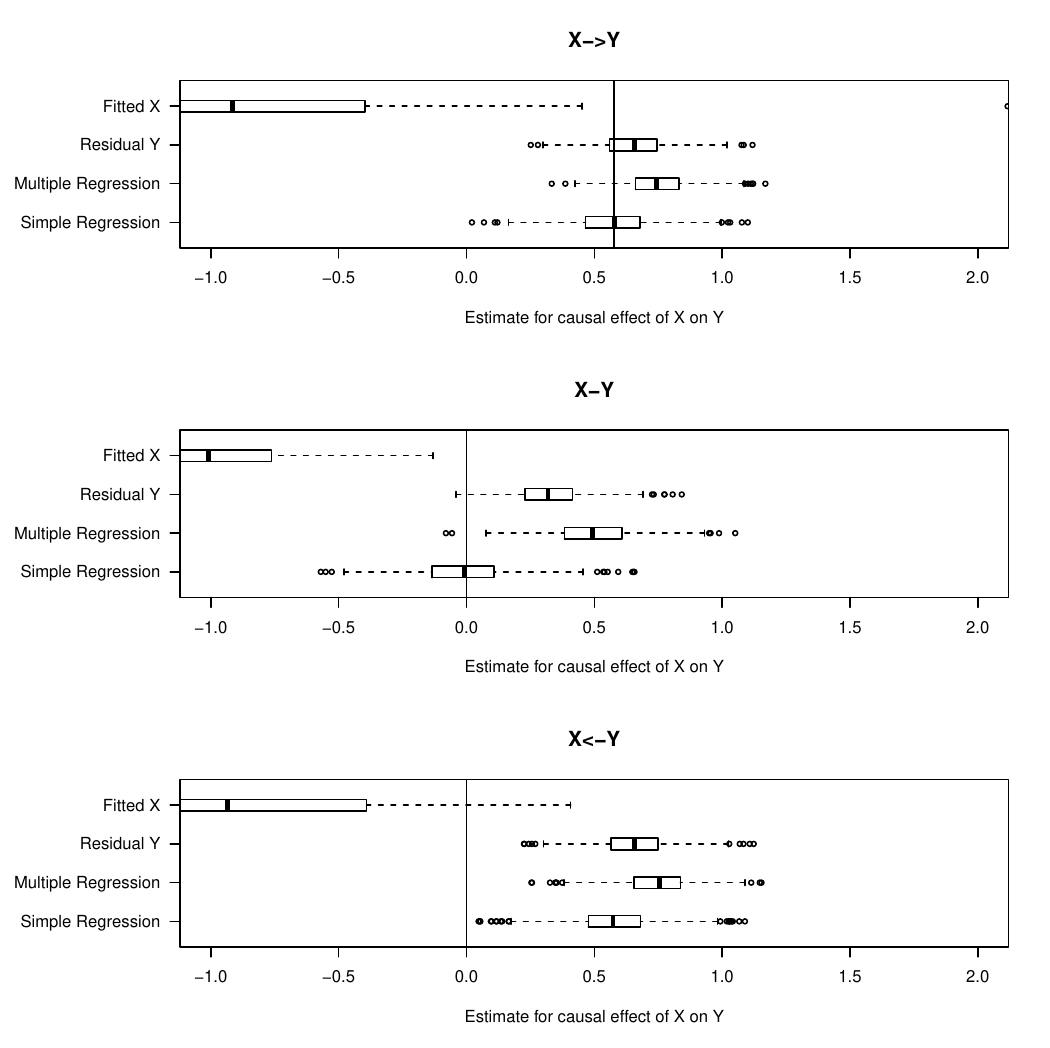}
\caption{$W$ as a twisted collider: $X\rightarrow^n W\leftarrow Y$.}
\label{TCol}
\end{figure}

\begin{figure}[p]
\centering
\includegraphics[width=5in]{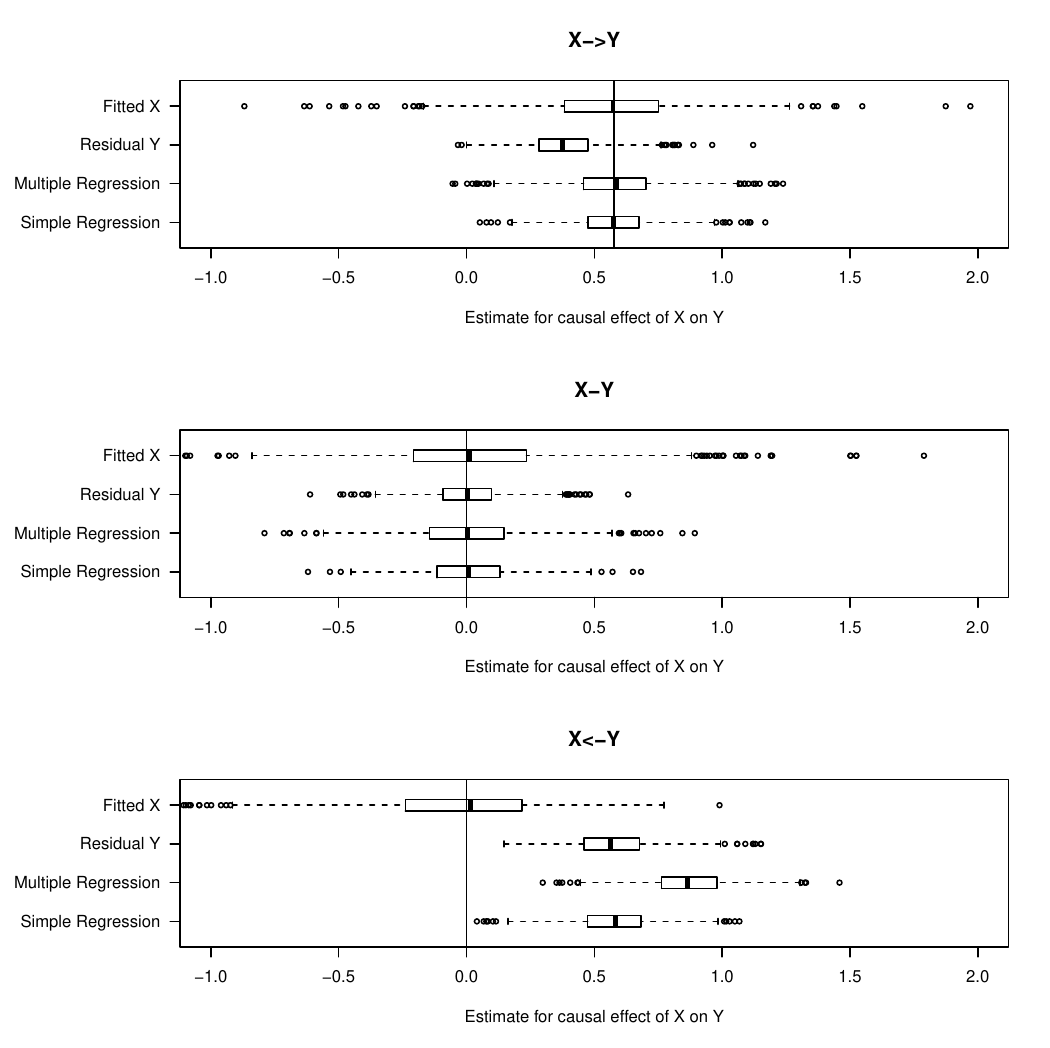}
\caption{Instrumental W: $X\leftarrow W-Y$.}
\label{Inst}
\end{figure}

\begin{figure}[p]
\centering
\includegraphics[width=5in]{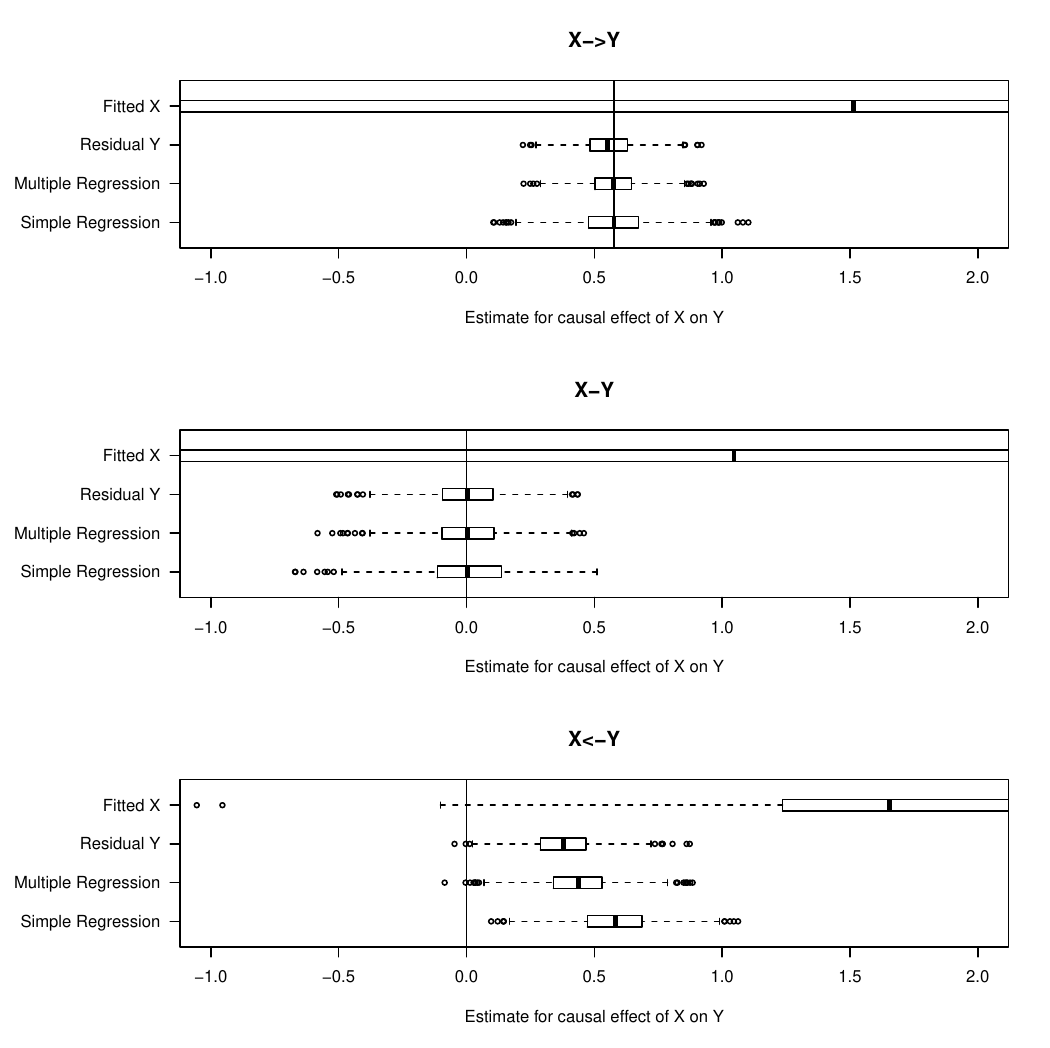}
\caption{Weak Confounding: $X-W\rightarrow Y$.}
\label{WConf}
\end{figure}

\begin{figure}[p]
\centering
\includegraphics[width=5in]{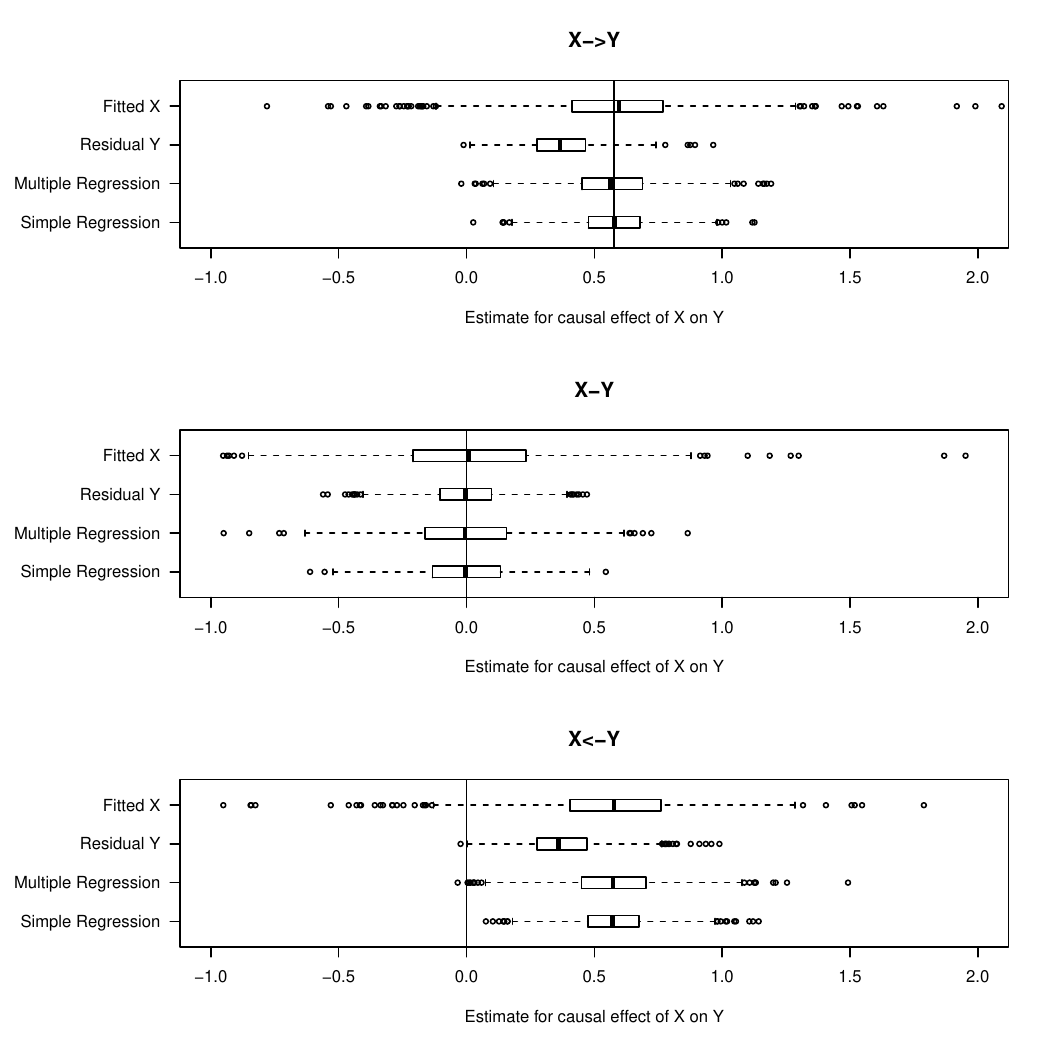}
\caption{$W$ Post Treatment: $X\rightarrow W-Y$.}
\label{PostT}
\end{figure}

\begin{figure}[p]
\centering
\includegraphics[width=5in]{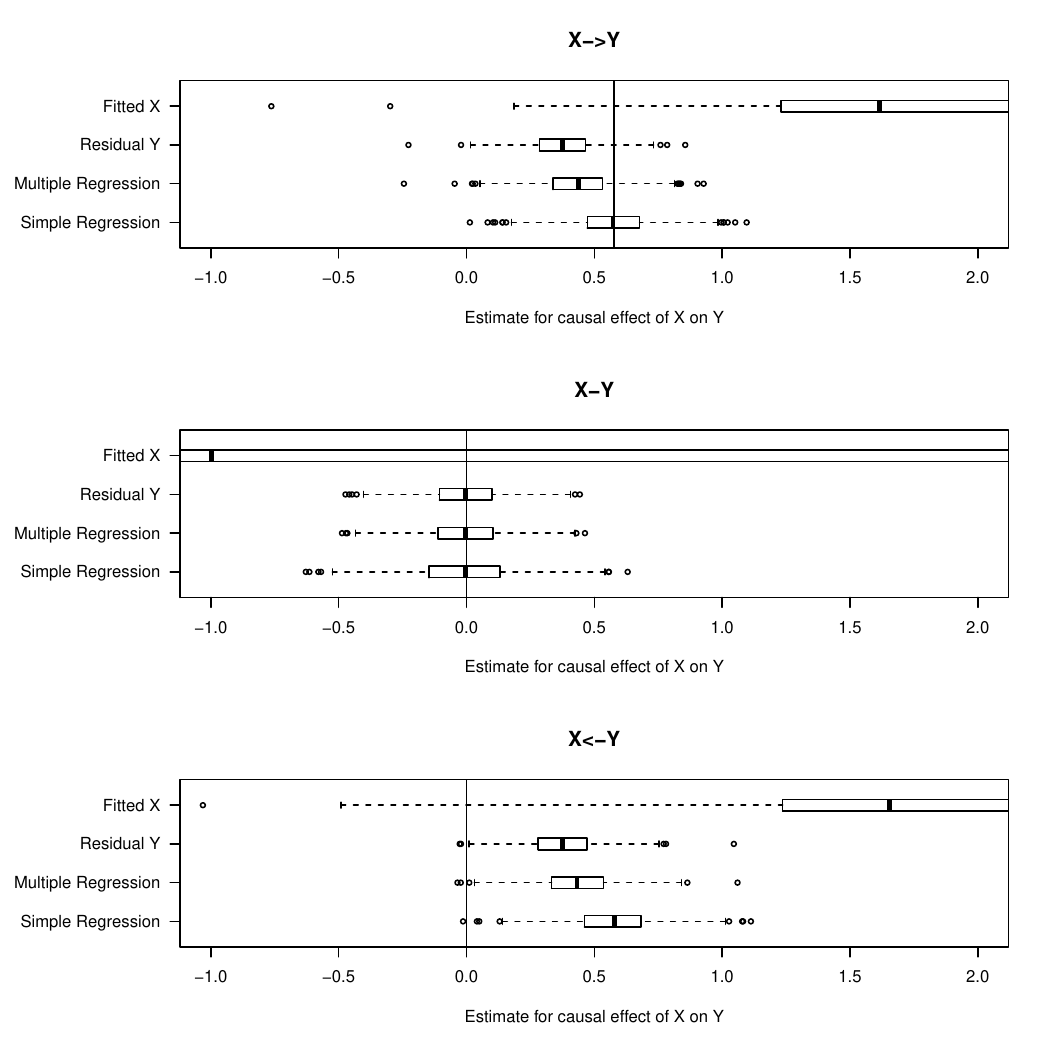}
\caption{$W$ Post Response: $X-W\leftarrow Y$.}
\label{PostR}
\end{figure}

\begin{figure}[p]
\centering
\includegraphics[width=5in]{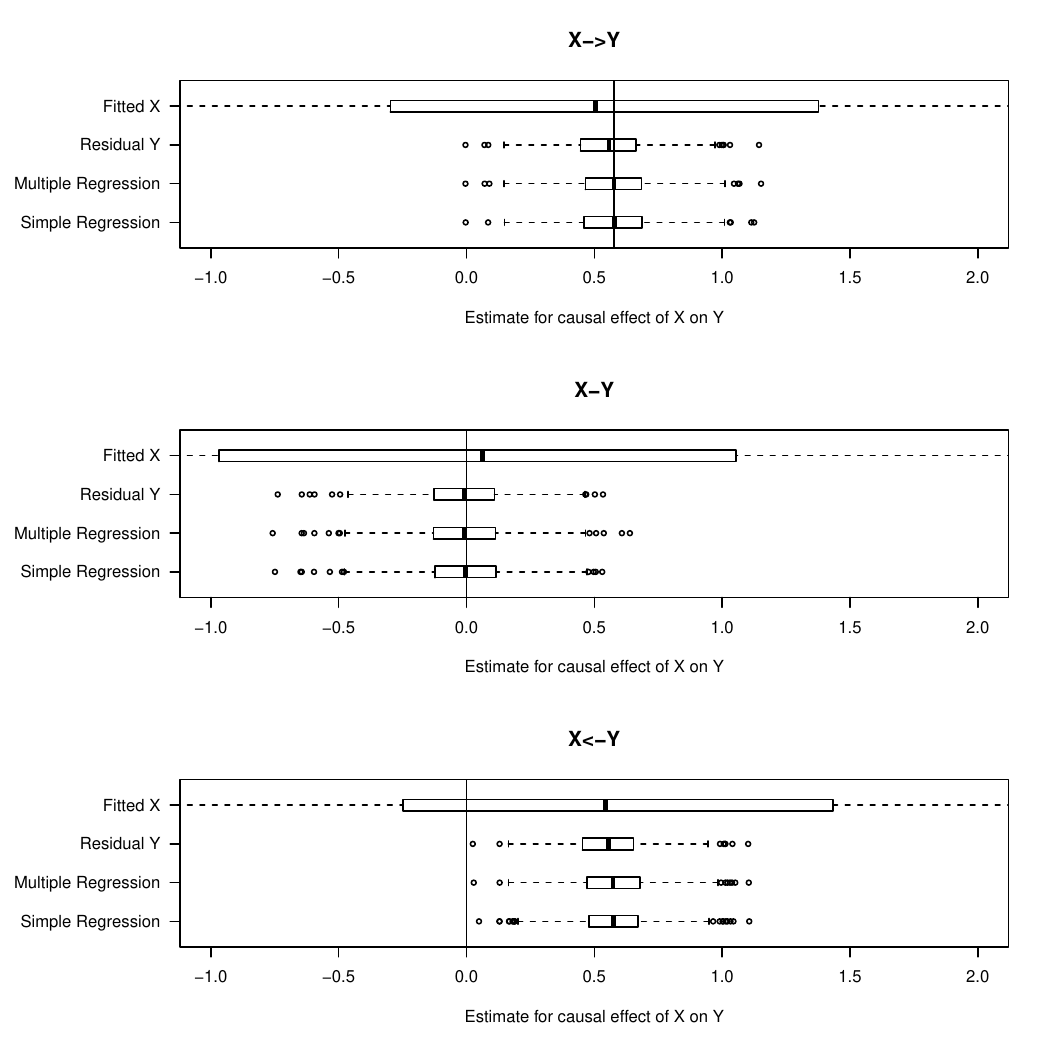}
\caption{$W$ Irrelevant: $X-W-Y$.}
\label{Irr}
\end{figure}

\begin{figure}[p]
\centering
\includegraphics[width=5in]{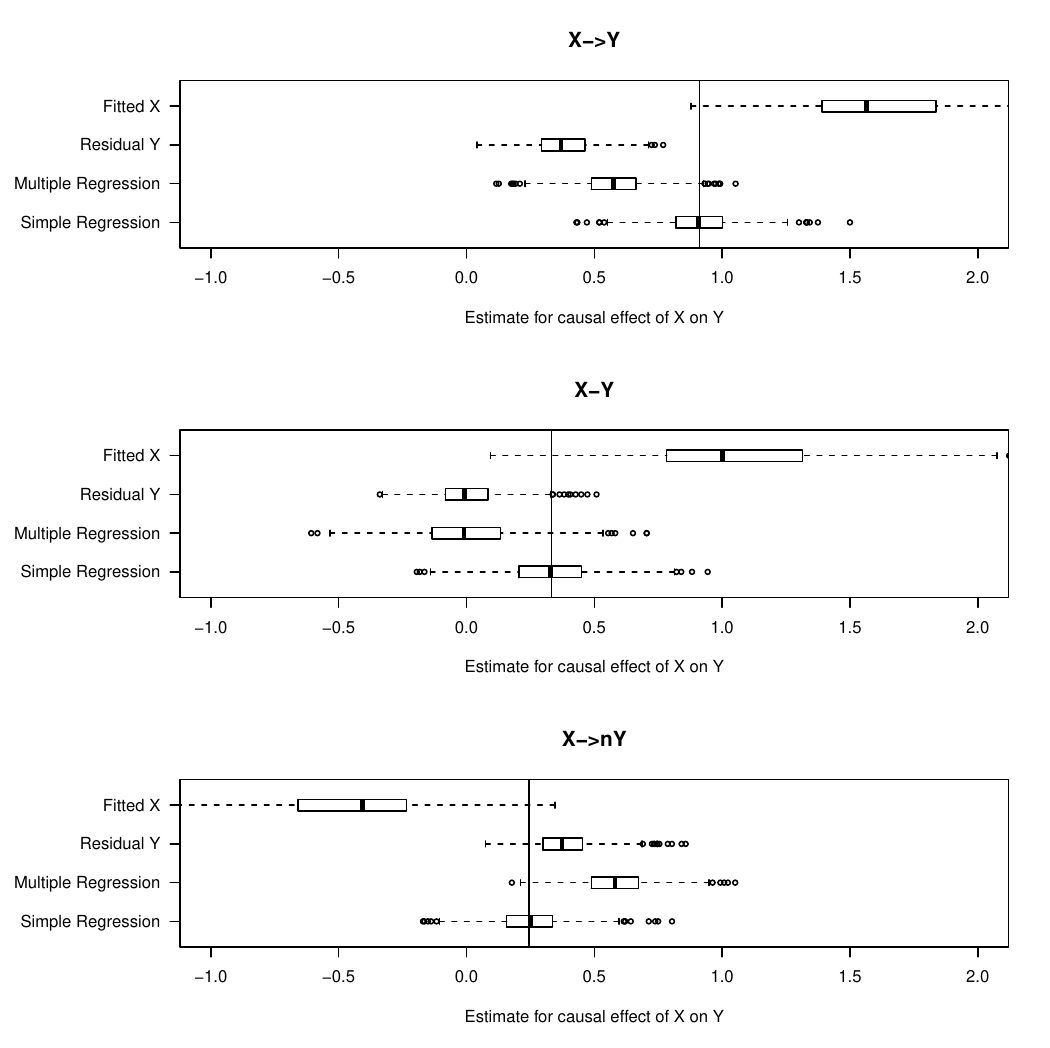}
\caption{Forward Chain: $X\rightarrow W\rightarrow Y$; cycle excluded, replaced by twisted forward chain.}
\label{FChain}
\end{figure}

\begin{figure}[p]
\centering
\includegraphics[width=5in]{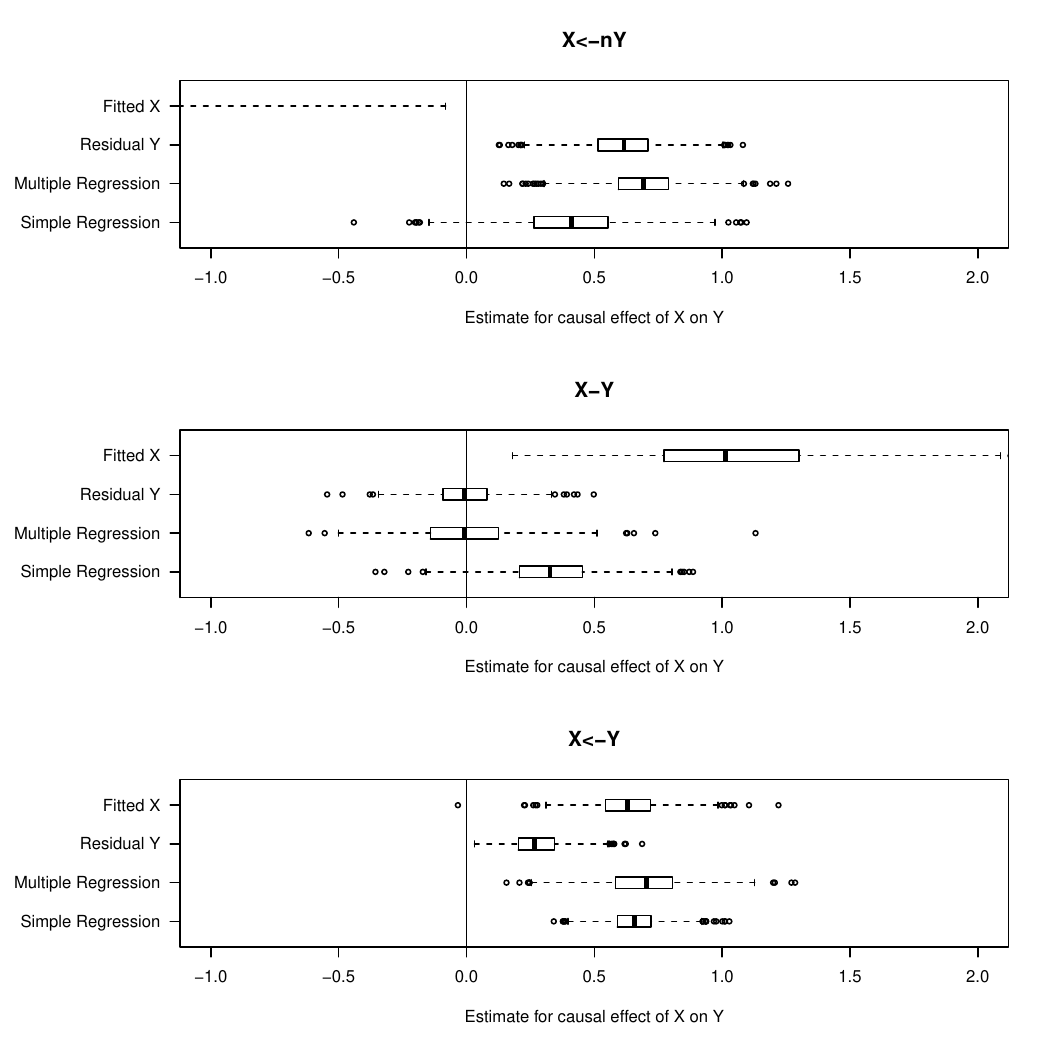}
\caption{Backward Chain: $X\leftarrow W\leftarrow Y$; cycle excluded, replaced by twisted backward chain.}
\label{BChain}
\end{figure}
\FloatBarrier

\newpage
\section{Comments}
With an unobserved confounder $U$ as shown here
\newline
 \begin{tikzpicture}
\node (x) at (0,1) {$X$};
\node (y) at (1,1) {$Y$};
\node (w) at (0,0) {$W$};
\node (u) at (1,0) {$U$};
 
\draw[->, line width= 1] (x) -- (y);
\draw[->, line width= 1] (w) -- (x);
\draw[->, line width= 1] (u) -- (x);
\draw[->, line width= 1] (u) -- (y);
\end{tikzpicture}  
we would expect Fitted $X$ to outperform Simple Regression \citep[Chapter 7]{Morgan07}.  Fitted $X$ matches the causal structure of Figure \ref{Inst}, where it is the only consistently unbiased technique.  Fitted $X$ performs poorly otherwise.  

Residual $Y$ is surprisingly robust against a mis specified direction for the causal effect between $X$ and $Y$.  Take Figure \ref{Conf} as an example.  There we see for each plot that the bias of Residual $Y$ is bounded by approximately $0.3$.  This is a better bound than any of the alternatives.  When estimating the causal effect of $X$ on $Y$ a spurious result is possible when $Y$ causes $X$ in reality.  The risk for such a spurious result can be reduced in some situations by using Residual $Y$. 

Another observed benefit of Residual $Y$ is its generally small variance.  With a large sample and a known causal graph, theory says \citep{Pearl09b,Vand} when $W\rightarrow X$ then Multiple Regression should be used in place of Simple Regression and when $X\rightarrow W$ then Simple Regression should be used instead of Multiple Regression, and this is consistent with our results.  Our results show the utility of additional adjustment techniques when samples are smaller or when there is uncertainty regarding the structure of the causal graph.   
\section{Discussion}
We have displayed results for sample size $n=30$ using box plots so as to show bias and sampling error.  Those interested in large sample asymptotic theory can focus on the medians of the box plots, while those interested in small sample application can focus on the width or span of the box plots.  We have assumed normality and linearity throughout.  In some cases real data may meet these assumptions, and in other cases the data may be transformed so as to meet these assumptions, but often data will not satisfy these assumptions.  Also, the data generating process may involve more than three variables, some or all possibly categorical.  The R code used for our simulations should be adapted and generalized before application to a specific analysis.  Our simple results are meant to guide intuition and spur new thought.

It is hard to believe that any proposed causal graph perfectly matches reality.  Given a set of candidate graphs, we can conduct simulations to assess the performance of various adjustment techniques across data generating processes associated with the candidate graphs, before selecting an adjustment technique in accordance with costs associated with our specific problem.  If costs rise dramatically only for very large errors then the conservative choice of Residual Y adjustment may be preferable to Multiple Regression adjustment, especially when reverse causality is possible.  On the other hand, when there is a high degree of confidence in a causal graph then the simulations may reveal a single best adjustment techniques that matches theoretical recommendations for very precise estimation.

\citet[p. 123]{Pearl09a} states that ``adjustment amounts to partitioning the population into groups that are homogeneous with respect to (the third variable)'', while \citet{Ding15} titled their paper ``To Adjust or Not  to Adjust?''.  Much concern \citep{Vand} has been given to the problem of determining which covariates to condition on.  \citet{Rubin09} argues against leaving observed covariates unbalanced.  It has been suggested that ignoring covariate data is contrary to Bayesianism \citep{GellmanBlog}.  \citet{Arm} shows how to use such covariate information.  Here we have simply presented the results of basic simulations in support of the notion that we should ask not only whether we should adjust or not, but also how we should best adjust.

\section{Appendix}
\begin{lemma}
The estimates for ``Residual $X$'' and ``Residual $X$ and $Y$'' are both identical to the estimate for ``Multiple Regression''.
\end{lemma}
\begin{proof}
Let $\bf{x}$,$\bf{w}$, and $\bf{y}$ be observed vectors of data associated with $n$ observations of $(X,W,Y)$.  We assume mean zero vectors of data without loss of generality so $\bf{x}$,$\bf{w}$, and $\bf{y}$ are within a linear, $(n-1)$-dimensional, subspace of $\mathbb{R}^n$ orthogonal to a vector of ones.  All that follows takes place within this subspace.  We assume also that $\mathbf{w}\neq k_1\mathbf{x}$ for any $k_1$. 

Let $\mathbf{p}_1$ be the projection of $\mathbf{y}$ onto the span of $\mathbf{x}$ and $\mathbf{w}$.   Let $\mathbf{p}_2$ be the projection of $\mathbf{x}$ onto $\mathbf{w}$, and note that $\mathbf{p}_2=k_2\mathbf{w}$ for some $k_2$.  Let $\mathbf{p}_3$ be the projection of $\mathbf{y}$ onto $\mathbf{w}$, and note that $\mathbf{p}_3=k_3\mathbf{w}$ for some $k_3$.

Let $\alpha_1$ be the Multiple Regression slope estimate associated with $X$.  By least-squares we have $\alpha_1\mathbf{x}+c_1\mathbf{w}=\mathbf{p}_1$ for some $c_1$.  Let $\mathbf{p}_4$ be the projection of $\mathbf{y}$ onto the span of $\mathbf{x}-\mathbf{p}_2$ and $\mathbf{w}$.  Note that $\mathbf{p}_4=\mathbf{p}_1$. 

Let $\alpha_2$ be the Residual $X$ slope estimate.  We have $\alpha_2(\mathbf{x}-\mathbf{p}_2))+c_2\mathbf{w}=\mathbf{p}_4$, which can be solved with $\alpha_2=\alpha_1$ and $c_2=c_1+\alpha_1k_2$.

Let $\alpha_3$ be the Residual $X$ and $Y$ slope estimate.  Note that $\mathbf{p}_1-\mathbf{p}_3$ is the projection of $\mathbf{y}-\mathbf{p}_3$ onto the span of $\mathbf{x}-\mathbf{p}_2$ and $\mathbf{w}$.  We have $\alpha_3(\mathbf{x}-\mathbf{p}_2))+c_3\mathbf{w}=\mathbf{p}_1-\mathbf{p}_3$, which (since $\mathbf{p}_3=k_3\mathbf{w}$) can be solved with $\alpha_3=\alpha_2$ and $c_3=c_2+k_3$.

\end{proof}
\begin{verbatim}
### R Program for simulations 
k=1000 ### number of samples
v1=numeric(k) 
v2=numeric(k)
v3=numeric(k)
v4=numeric(k) 
for (i in 1:k) {
### insert particular data generating process below
x=rnorm(30,0,1)
w=(1/sqrt(3))*x+rnorm(30,0,sqrt(2)/sqrt(3))
y=(1/sqrt(3))*w+(1/sqrt(3))*x+rnorm(30,0,1/sqrt(3))
### the above example is for a forward chain
### the true value is (1/sqrt(3))*(1/sqrt(3))+(1/sqrt(3))
f=lm(x~w)$fitted.values
g=lm(y~w)$residuals
v1[i]=summary(lm(y~x))$coefficients[2,1] 
v2[i]=summary(lm(y~x+w))$coefficients[2,1] 
v3[i]=summary(lm(g~x))$coefficients[2,1]
v4[i]=summary(lm(y~f))$coefficients[2,1]
}
### R program for graphing box plots
par(las=1)
par(mar=c(5, 9, 4, 2) + 0.1)
boxplot(v1,v2,v3,v4,
xlab=c("Estimate for causal effect of X on Y"),
horizontal=T,
names=c(v1="Simple Regression",v2="Multiple Regression",
v3="Residual Y",v4="Fitted X"),
boxwex=.3)
abline(v=(1/sqrt(3))*(1/sqrt(3))+(1/sqrt(3)))
\end{verbatim}

\end{document}